\documentclass[12pt,draftcls,onecolumn]{IEEEtran}
\usepackage{savesym}
\usepackage{amsmath}
\savesymbol{iint}
\usepackage{txfonts}
\restoresymbol{TXF}{iint}
\usepackage{txfonts}
\usepackage{amsfonts}
\usepackage{amssymb}
\usepackage{amsmath}
\usepackage{mathrsfs}
\usepackage{amsmath,amssymb}
\usepackage{amsmath}
\usepackage{graphicx}

\newtheorem{remark}{Remark}

\newtheorem{theorem}{Theorem}
\newtheorem{definition}{Definition}
\newtheorem{lemma}{Lemma}
\begin{document}
%
\title{Graph-Theoretic Characterizations of Structural Controllability for
Multi-Agent System with Switching Topology*}
%
%
%

\author{Xiaomeng Liu$^{a}$, Hai
Lin$^{b}$$^{\ast}$,~\IEEEmembership{Senior Member,~IEEE,} and Ben M.
Chen$^{a}$, ~\IEEEmembership{Fellow,~IEEE}
\thanks{$^\ast$Corresponding author. Email:
hlin1@nd.edu  tel 574-6313177 fax 574-6314393}
\thanks{$^{a}${\em{Dept. of
ECE, National University of Singapore, Singapore}}~~$^{b}${\em{Dept
of EE, Univ of Notre Dame, Notre Dame, IN 46556, USA}}}
}

%
%

\markboth{}%
{Xiaomeng Liu \MakeLowercase{\textit{et al.}}: Structural
Controllability of Switched Linear Systems}
%



\maketitle

\begin{abstract}                          
This paper considers the controllability problem for multi-agent
systems. In particular, the structural controllability of
multi-agent systems under switching topologies is investigated. The
structural controllability of multi-agent systems is a
generalization of the traditional controllability concept for
dynamical systems, and purely based on the communication topologies
among agents. The main contributions of the paper are
graph-theoretic characterizations of the structural controllability
for multi-agent systems. It turns out that the multi-agent system
with switching topology is structurally controllable if and only if
the union graph $\mathcal{G}$ of the underlying communication
topologies is connected (single leader) or leader-follower connected
(multi-leader). Finally, the paper concludes with several
illustrative examples and discussions of the results and future
work.
\end{abstract}

\section{Introduction}
Due to the latest advances in communication and computation, the
distributed control and coordination of the networked dynamic agents
has rapidly emerged as a hot multidisciplinary research area
\cite{lb}\cite{dm}\cite{bf}, which lies at the intersection of
systems control theory, communication and mathematics.
In addition, the advances of the research in multi-agent systems are
strongly supported by their promising civilian and military
applications, such as cooperative control of unmanned air
vehicles(UAVs), autonomous underwater vehicles(AUVs), space
exploration, congestion control in communication networks, air
traffic control and so on \cite{tp}\cite{hkk}. Much work has been
done on the formation stabilization and consensus seeking.
Approaches like graph Laplacians for the associated neighborhood
graphs, artificial potential functions, and navigation functions for
distributed formation stabilization with collision avoidance
constraints have been developed. Furthermore, inspired by the
cooperative behavior of natural swarms, such as bee flocking, ant
colonies and fish schooling, people try to obtain experiences from
how the group units make their whole group motions under control
just through limited and local interactions among them.

The control of such large scale complex systems poses several new
challenges that fall beyond the traditional methods. Part of the
difficulty comes from the fact that the global behavior of the whole
group combined by multiple agents is not a simple summation of the
individual agent's behavior. Actually, the group behavior can be
greatly impacted by the communication protocols or interconnection
topology between the agents, which makes the global behavior display
high complexities. Hence, the cooperative control of multi-agent
systems is still in its infancy and attracts more and more
researchers' attention. One basic question in multi-agent systems
that attracts control engineers' interest is what is the necessary
information exchanging among agents to make the whole group
well-behaved, e.g., controllable. This can be formulated as a
controllability problem for multi-agent systems under the
leader-follower framework. Roughly speaking, a multi-agent system is
controllable if and only if we can drive the whole group of agents
to any desirable configurations only based on local interactions
between agents and possibly some limited commands to a few agents
that serve as leaders. The basic issue is the interplay between
control and communication. In particular, we would like to
investigate what is the necessary and/or sufficient condition on the
graph of communication topologies among agents for the
controllability of multi-agent systems.

This multi-agent controllability problem was first proposed in
\cite{ht}, which formulated it as the controllability of a linear
system and proposed a necessary and sufficient algebraic condition
based on the eigenvectors of the graph Laplacian. Reference
\cite{ht} focused on fixed topology situation with a particular
member which acted as the single leader. The problem was then
developed in \cite{mj}\cite{mj2}\cite{Raha}\\\cite{zji}, and got
some interesting results. For example, in \cite{mj2}, it was shown
that a necessary and sufficient condition for controllability is not
sharing any common eigenvalues between the Laplacian matrix of the
follower set and the Laplacian matrix of the whole topology.
However, it remains elusive on what exactly the graphical meaning of
these algebraic conditions related to the Laplacian matrix. This
motivates several research activities on illuminating the
controllability of multi-agent systems from a graph theoretical
point of view. For example, a notion of anchored systems was
introduced in \cite{Raha} and it was shown that symmetry with
respect to the anchored vertices makes the system uncontrollable.
However, so far the research progress using graph theory is quite
limited and a satisfactory graphical interpretation of these
algebraic controllability conditions turns out to be very
challenging. Besides, the weights of communication linkages among
agents have been demonstrated to have a great influence on the
behavior of whole multi-agent group (see \cite{m} as an example).
However, in the previous multi-agent controllability literature
\cite{ht}-\cite{Raha}, the communication weighting factor is usually
ignored. One classical result under this no weighting assumption is
that a multi-agent system with complete graphical communication
topology is uncontrollable \cite{ht}. This is counter-intuitive
since it means each agent can get direct information from each other
but this leads to a bad global behavior as a team. This shows that
too much information exchange may damage the controllability of
multi-agent system. In contrast, if we set weights of unnecessary
connections to be zero and impose appropriate weights to other
connections so as to use the communication information in a
selective way, then it is possible to make the system controllable
\cite{mz}.

In this paper, motivated by the above observation, the weighting
factor is taken into account for multi-agent controllability
problem. In particular, a new notion for the controllability of
multi-agent systems, called structural controllability, which was
proposed by us in \cite{mz}, is investigated directly through the
graph-theoretic approach for control systems. The communication
topology of whole multi-agent system is described by a weighted
graph and the system is called structurally controllable if one may
find a set of weights such that the corresponding multi-agent system
is controllable in a classical sense. The structural controllability
reveals under certain topology whether it is possible to make the
whole multi-agent system well-behaved, i.e. controllable here
through suitable choice of the communication weights. From another
angle, it helps to bring to light the effects of the communication
topology on the controllability property of multi-agent systems
without worrying about the influence of weighting factors. It turns
out that this controllability notation only depends on the topology
of the communication scheme in the case of a single leader under a
fixed topology \cite{mz}.

Another novelty in this paper is the successful investigation of
impacts of switching topologies on the multi-agent controllability
property, for which there is barely graphical interpretation to the
best of our knowledge. Note that the results in
\cite{ht}\cite{mj}\cite{mj2}\\\cite{Raha}\cite{zji}\cite{mz} are all
focused on multi-agent systems under fixed communication topologies
which may restrict their impacts on real applications. 
In many applications, it may become impossible to keep the
communication topology fixed for the whole period. Therefore, it is
of practical importance to consider time varying communication
topologies. A natural framework to study the time variance of
communication topology is through switched systems, see e.g.,
\cite{linhai}\cite{linhai2}\cite{sunzd}. In this paper, we will
focus on multi-agent systems under switching topologies in the
framework of switched systems. Some early efforts have been observed
in the literature. Necessary and sufficient algebraic conditions for
the controllability of multi-agent systems under switching topology
were derived in \cite{zji2}\cite{liub} based on the developments of
controllability study in switched systems. However, these algebraic
results lacks graphically intuitive interpretations, which are
important since they can provide us significant guidelines for the
communication protocol design for multi-agent systems. Therefore,
this paper aims to fill this gap and propose graphic interpretations
of these algebraic conditions for the controllability of multi-agent
systems under switching topology. In particular, we follow the setup
in \cite{mz} and investigate the structural controllability of
multi-agent systems with switching communication topologies and for
both single leader and multi-leader cases. It is assumed that the
leaders act as the external or control signal and will not be
affected by any other group members. Based on this structural
controllability, we propose necessary and sufficient graph theoretic
conditions for the structural controllability of multi-agent system
with switching topologies. It turns out that the multi-agent system
with switching topology is structurally controllable if and only if
the union graph $\mathcal{G}$ of the underlying communication
topologies is connected (single leader) or leader-follower connected
(multi-leader). Some examples are given to underscore our
theoretical analysis.

The outline of this paper is as follows: In Section II, we introduce
some basic preliminaries and the problem formulation, followed by
structural controllability study in Section III, where a graphic
necessary and sufficient condition for the structural
controllability under single leader case is given. In Section IV,
graphical interpretation of structural controllability of
multi-leader multi-agent system is proposed. In Section V, some
examples are presented to give the readers deeper understanding of
our theoretical results. Finally, some concluding remarks are drawn
in the paper.
\section{Preliminaries and Problem Formulation}
\subsection{Graph Theory Preliminaries}
 A weighted graph is an appropriate representation for the communication or sensing links
among agents because it can represent both existence and strength of
these links among agents. The weighted graph $\mathcal{G}$ with $N$
vertices consists of a vertex set $\mathcal{V} = \{ v_1 ,v_2 ,
\ldots ,v_N \} $ and an edge set $\mathcal{I}= \{ e_1 ,e_2 , \ldots
,e_{N'} \} $, which is the interconnection links among the vertices.
Each edge in the weighted graph represents a bidirectional
communication or sensing media. Two vertices are known to be
\textit{neighbors} if $(i,j)\in \mathcal{I}$, and the number of
neighbors for each vertex is its \textit{valency}. An alternating
sequence of distinct vertices and edges in the weighted graph is
called a \textit{path}. The weighted graph is said to be
\textit{connected} if there exists at least one path between any
distinct vertices, and \textit{complete} if all vertices are
neighbors to each other.

The \textit{adjacency matrix}, $\mathcal{A}$ is defined as
\begin{equation}
\mathcal{A}_{(i,j)}=\left\{\begin{array}{cc}
  \mathcal{W}_{ij}&(i,j)\in \mathcal{I}, \hfill  \\
   0&otherwise,\hfill  \\
\end{array}\right.
\end{equation}
where $\mathcal{W}_{ij}$ $\neq$ 0 stands for the weight of edge
$(j,i)$. Here, the adjacency matrix $\mathcal{A}$ is $\mid
\mathcal{V}\mid\times\mid\mathcal{V}\mid$ and $\mid.\mid$ is the
cardinality of a set.

The \textit{Laplacian matrix} of a graph $\mathcal{G}$, denoted as
$\mathcal{L(G)} \in \mathbb{R}^{\mid
\mathcal{V}\mid\times\mid\mathcal{V}\mid}$ or $\mathcal{L}$ for
simplicity, is defined as
\begin{equation}
\mathcal{L}_{(i,j)}=\left\{\begin{array}{cc}
  \Sigma_{j\in \mathcal{N}_i}~\mathcal{W}_{ij}&i=j,  \\
  -\mathcal{W}_{ij}&i\neq j~and~(i,j)\in \mathcal{I}, \\
   0&otherwise. \\
\end{array}\right.
\end{equation}
\subsection{Multi-agent Structural Controllability with Switching Topology}
Specifically, controllability problem usually cares about how to
control $N$ agents based on the leader-follower framework. Take the
case of single leader as example. Without loss of generality, assume
that the $N$-th agent serves as the leader and takes commands and
controls from outside operators directly, while the rest $N-1$
agents are followers and take controls as the nearest neighbor law.

Mathematically, each agent's dynamics can be seen as a point mass
and follows
\begin{equation}\label{eq11}
 \dot{x}_i=u_i.
\end{equation}
The control strategy for driving all follower agents is
\begin{equation}\label{eq12}
 u_i  =  - \sum\nolimits_{j \in \mathcal{N}_i} {w_{ij}(x_i -
x_j )}+w_{ii}x_i,
\end{equation}
where $\mathcal{N}_i$ is the neighbor set of the agent $i$ (could
contain $i$ itself), and $w_{ij}$ is weight of the edge from  agent
$j$ to agent $i$. 
On the other hand, the leader's control signal is not influenced by
the followers and needs to be designed, which can be represented by
\begin{equation}
\dot x_N = u_N.
\end{equation}

In other words, the leader affects its nearby agents, but it does
not get directly affected by the followers since it only accepts the
control input from an outside operator. For simplicity, we will use
$z$ to stand for $x_N$ in the sequel. It is known that the whole
multi-agent system under fixed communication topology can be written
as a linear system:
\begin{equation}\label{eq13}
\begin{array}{*{20}c}
   {\left[ {\begin{array}{*{20}c}
   {\dot x}  \\
   {\dot z}  \\
\end{array}} \right] = \left[ {\begin{array}{*{20}c}
   A & B  \\
   0 & 0  \\
\end{array}} \right]\left[ {\begin{array}{*{20}c}
   x  \\
   z  \\
\end{array}} \right] + \left[ {\begin{array}{*{20}c}
   0  \\
   {u_N }  \\
\end{array}} \right]} & {}
\end{array},
\end{equation}
where $A \in \mathbb{R}^{(N-1) \times (N-1)}$ and $B\in
\mathbb{R}^{(N-1) \times 1}$ are both sub-matrices of the
corresponding graph Laplacian matrix $-\mathcal{L}$ 


The communication network of dynamic agents with directed
information flow under link failure and creation can usually
described by switching topology. Under $m$ switching topologies, it
is clear that the whole system equipped with $m$ subsystems can be
written in a compact form
\begin{equation}\label{eq13}
\begin{array}{*{20}c}
   {\left[ {\begin{array}{*{20}c}
   {\dot x}  \\
   {\dot z}  \\
\end{array}} \right] = \left[ {\begin{array}{*{20}c}
   A_i & B_i  \\
   0 & 0  \\
\end{array}} \right]\left[ {\begin{array}{*{20}c}
   x  \\
   z  \\
\end{array}} \right] + \left[ {\begin{array}{*{20}c}
   0  \\
   {u_N }  \\
\end{array}} \right]} & {}
\end{array},
\end{equation}
or, equivalently,
\begin{equation}\label{eq14}
\left\{\begin{array}{clc}
   \dot x &=& A_i x + B_i z , \hfill  \\
   \dot z &=& u_N , \hfill  \\
\end{array}\right.
\end{equation}
where $i\ \in \{1,\ldots,m\}$. $A_i \in \mathbb{R}^{(N-1) \times
(N-1)}$ and $B_i\in \mathbb{R}^{(N-1) \times 1}$ are both
sub-matrices of the corresponding graph Laplacian matrix
$-\mathcal{L}$. The matrix $A_i$ reflects the interconnection among
followers, and the column vector $B_i$ represents the relation
between followers and the leader under corresponding subsystems.
Since the communication topologies among agents are time-varying, so
the matrices $A_i$ and $B_i$ are also varying as a function of time.
Therefore, the dynamical system described in (\ref{eq13}) can be
naturally modeled as a switched system (definition can be found
latter).

Considering the structural controllability of multi-agent system,
system matrices $A_i$ and $B_i$, $i\ \in \{1,\ldots,m\}$ are
structured matrices, which means that their elements are either
fixed zeros or free parameters. Fixed zeros imply that there is no
communication link between the corresponding agents and the free
parameters stand for the weights of the communication links. Our
main task here is to find out under what kinds of communication
topologies, it is possible to make the group motions under control
and steer the agents to the specific geometric positions or
formation as a whole group. Now this controllability problem reduces
to whether we can find a set of weights $w_{ij}$ such that the
multi-agent system (\ref{eq13}) is controllable. 
Then the controllability problem of multi-agent system can now be
formulated as the structural controllability problem of switched
linear system (\ref{eq13}):
\begin{definition}\label{def9} The multi-agent system (\ref{eq13}) with switching topology, whose matrix elements are zeros or undetermined parameters, is
said to be structurally controllable if and only if there exist a
set of communication weights $w_{ij}$ that can make the system
(\ref{eq13}) controllable in the classical sense.
\end{definition}
\subsection{Switched Linear System and Controllability Matrix}

In general, a switched linear system is composed of a family of
subsystems and a rule that governs the switching among them, and is
mathematically described by
\begin{eqnarray}\label{eq2}
\dot x(t)=&A_{\sigma(t)} x(t)+B_{\sigma(t)} u(t) ,
\end{eqnarray}
where $x(t)\in \mathbb{R}^n$ are the states, $u(t)\in
\mathbb{R}^{r}$, are piecewise continuous input,
$\sigma:[t_0,\infty)\rightarrow M \triangleq \{1,\ldots,m\}$ is the
switching signal. System (\ref{eq2}) contains $m$ subsystems
$(A_i,B_i),$ $i\in \{1,\ldots,m\}$ and $\sigma(t)$= $i$ implies that
the $i$th subsystem $(A_i,B_i)$ is activated at time instance $t$

For the controllability problem of switched linear systems, a
well-known matrix rank condition was given in \cite{ZS}:

\begin{lemma}\label{lem1}(\cite{ZS}) If matrix:\begin{equation}\begin{split}\label{eq3} &[B_1,\ldots,B_m,
A_1B_1,\ldots,A_mB_1,\ldots,A_mB_m,A_1^2B_1,\ldots,A_mA_1B_1,\\&\ldots,A_1^2B_m,\ldots,
A_mA_1B_m,\ldots,A_1^{n-1}B_1,\ldots,A_mA_1^{n-2}B_1,\ldots,\\&A_1A_m^{n-2}B_m,\ldots,
A_m^{n-1}B_m] \end{split}\end{equation}\\ has full row rank $n$,
then switched linear system (\ref{eq2}) is controllable, and vice
versa.
\end{lemma}

This matrix is called the controllability matrix of the
corresponding switched linear system (\ref{eq2}).

\section{Structural Controllability of Multi-agent System with Single Leader }

The multi-agent system with a single leader under switching topology
has been modeled as switched linear system (\ref{eq13}). Before
proceeding to the structural controllability study, we first discuss
the controllability of multi-agent system (\ref{eq13}) when all the
communication weights are fixed.

After simple calculation, the controllability matrix of switched
linear system (\ref{eq13}) can be shown to have the following form:
\begin{equation}\left[ {\begin{array}{*{20}c}
  \label{eq3} 0,&\ldots&,0,&B_1,&\ldots,&B_m,&A_1B_1&
\ldots,&A_1A_m^{N-3}B_m,&\ldots,&
A_m^{N-2}B_m \\
  1,&\ldots&,1,& 0, & 0, & 0,& 0,& 0,&0,&0,&0\\
\end{array}} \right]\nonumber\end{equation}
This implies that the controllability of the system (\ref{eq13})
coincides with the controllability of the following system:
\begin{equation}\label{eq6}
\begin{array}{clc}
   \dot x &=& A_i x + B_i z \ \ \ \ \ \ i\ \in \{1,\ldots,m\} \\
   \end{array}.
\end{equation}
Which is the extracted dynamics of the followers that correspond to
the $x$ component of the equation. Therefore,
\begin{definition}\label{def10}The multi-agent system (\ref{eq13}) is said to be structurally controllable
under leader $z$ if system (\ref{eq6}) is structurally controllable
under control input $z$.
\end{definition} For simplicity, we use $(A_i,B_i)~~i\ \in \{1,\ldots,m\}$ to
represent switched linear system (\ref{eq6}) in the sequel.

In (\ref{eq6}), each subsystem $(A_i, B_i)$ can be described by a
directed graph \cite{lin}:

\begin{definition}The representation graph of structured system $(A_i, B_i)$ is a directed graph
$\mathcal{G}_i$, with vertex set $\mathcal{V}_i=\mathcal{X}_i\cup
\mathcal{U}_i$, where $\mathcal{X}_i=\{x_1,x_2,\ldots,x_n\}$, which
is called $state~vertex~ set$ and
$\mathcal{U}_i=\{u_1,u_2,\dots,u_r\}$, which is called
$input~vertex~ set$, and edge set
$\mathcal{I}_i=\mathcal{I}_{U_iX_i}\cup\mathcal{I}_{X_iX_i}$, where
$\mathcal{I}_{U_iX_i}=\{(u_p,x_q)|B_{qp}\neq 0, 1\leq p \leq r,
1\leq q\leq n\}$ and $\mathcal{I}_{X_iX_i}=\{(x_p,x_q)|A_{qp}\neq 0,
1\leq p \leq n, 1\leq q\leq n\}$ are the oriented edges between
inputs and states and between states defined by the interconnection
matrices $A_i$ and $B_i$ above. This directed graph (for notational
simplicity, we will use digraph to refer to directed graph)
$\mathcal{G}_i$ is also called the graph of matrix pair $(A_i,B_i)$
and denoted by $\mathcal{G}_i(A_i,B_i)$.
\end{definition}

For each subsystem, we have got a graph $G_i$ with vertex set
$\mathcal{V}_i$ and edge set $\mathcal{I}_i$ to represent the
underlaying communication topologies. As to the whole switched
system, the representing graph, which is called $union~graph$, is
defined as follows:

\begin{definition} The switched linear system (\ref{eq6}) can be represented by a union digraph,
defined as a flow structure $\mathcal{G}$. Mathematically,
$\mathcal{G}$ is defined as
\begin{eqnarray}
\mathcal{G}_1\cup \mathcal{G}_2\cup\ldots\cup
\mathcal{G}_m=\{\mathcal{V}_1\cup \mathcal{V}_2\cup \ldots \cup
\mathcal{V}_m;\\\mathcal{I}_1\cup \mathcal{I}_2\cup \ldots\cup
\mathcal{I}_m\}\nonumber \end{eqnarray} \end{definition}

\begin{remark}\label{rem3}It turns out that union
graph $\mathcal{G}$ is the representation of linear structured
system: $(A_1+A_2+\ldots+A_m,B_1+B_2+\ldots+B_m)$. \end{remark}

\begin{remark}In lots of literature about controllability of multi-agent systems
\cite{ht}-\cite{liub}, the underlying communication topology among
the agents is represented by undirected graph, which means that the
communication among the agents is bidirectional. Here we still adopt
this kind of communication topology. Then $w_{ij}$ and $w_{ji}$ are
free parameters or zero simultaneously (in numerical realization,
the values of $w_{ij}$ and $w_{ji}$ can be chosen to be different).
Besides, one edge in undirected graph can be treated as two oriented
edges. Consequently, even though all the analysis and proofs for
structural controllability of multi-agent systems are based on the
directed graph (the natural graphic representation of matrix pair
$(A_i,B_i)$ is digraph), the final result will be expressed in
undirected graph form.
\end{remark}

Before proceeding further, we need to introduce two definitions
which were proposed in \cite{lin} for linear structured system $\dot
x=Ax+Bu$ first:

\begin{definition}\label{def5}(\cite{lin}) The matrix pair $(A,B)$ is said to be
reducible or of form I if there exist permutation matrix $P$ such
that they can be written in the following form:
\begin{eqnarray}
PAP^{-1}=\left[
\begin{array}{ccc}
A_{11}&0\\
A_{21}&A_{22}\\
\end{array}\right],PB=\left[
\begin{array}{cc}
0\\
B_{22}\\
\end{array}\right],
\end{eqnarray}\\ where $ A_{11}\in \mathbb{R}^{p \times p}$ , $A_{21 }
\in \mathbb{R}^{(n - p) \times p}$,$ $ $  A_{22}  \in \mathbb{R}^{(n
- p) \times (n- p)} $ and $ B_{22}  \in \mathbb{R}^{(n- p) \times
r}$.
\end{definition}
\begin{remark}\label{rem4} Whenever the matrix pair $(A,B)$ is of form I, the system is
structurally uncontrollable \cite{lin} and meanwhile, the
controllability matrix $Q=\left[B,AB,\ldots,A^{n-1}B\right]$ will
have at least one row which is identically zero for all parameter
values \cite{KL}. If there is no such permutation matrix $P$, we say
that the matrix pair $(A,B)$ is irreducible.
\end{remark}

\begin{definition}\label{def6} (\cite{lin}) The matrix pair $(A,B)$ is said to be of
form II if there exist permutation matrix $P$ such that they can be
written in the following form:
\begin{eqnarray}
\left[PAP^{-1},PB\right]=\left[
\begin{array}{ccc}
P_1\\
P_2
\end{array}\right],
\end{eqnarray}\\ where $P_2\in \mathbb{R}^{(n-k)\times (n+r)}$ ,
$P_1 \in \mathbb{R}^{k \times (n+r)}$ with no more than $k-1$
nonzero columns (all the other columns of $P_1$ have only fixed zero
entries).
\end{definition}
%


%
%
%

Here we need to recall a known result in literature for structural
controllability of multi-agent system with fixed topology \cite{mz}:

\begin{lemma}\label{lem2} The multi-agent system with fixed topology under the
communication topology $\mathcal{G}$ is structurally controllable if
and only if graph $\mathcal{G}$ is connected.
\end{lemma}

This lemma proposed an interesting graphic condition for structural
controllability in fixed topology situation and revealed that the
controllability is totally determined by the communication topology.
However, how about in the switching topology situation? According to
lemma \ref{lem1}, once we impose proper scalars for the parameters
of the system matrix $(A_i,B_i)$ to satisfy the full rank condition,
the multi-agent system (\ref{eq6}) is structurally controllable.
However, this only proposed an algebraic condition. Do we still have
very good graphic interpretation for the relationship between the
structural controllability and switching interconnection topologies?
The following theorem answers this question and gives a graphic
necessary and sufficient condition for structural controllability
under switching topologies.
\begin{theorem}\label{the1}The multi-agent system (\ref{eq6}) with the communication
topologies $\mathcal{G}_i$, $i\in \{1,\ldots,m\}$ is structurally
controllable if and only if the union graph $\mathcal{G}$ is
connected.
\end{theorem}

\begin{proof}\textit{~~Necessity:} Assume that the multi-agent switched
system is structurally controllable, we want to prove that the union
graph $\mathcal{G}$ is connected, which is equivalent to that the
system has no isolated agents in the union graph $\mathcal{G} $
\cite{mz}.

Suppose that the union graph $\mathcal{G}$ is disconnected and for
simplicity, we will prove by contradiction for the case that there
exits only one disconnected agent. The proof can be
straightforwardly extended to more general cases with more than one
disconnected agents. If there is one isolated agent in the union
graph, there are two possible situations: the isolated agent is the
leader or one of the followers. On one hand, if the isolated agent
is the leader, it follows that $ B_1+B_2+\ldots+B_m$  is identically
a null vector. So every $ B_i$ is a null vector. Easily we can
conclude that the controllability matrix for the switched system is
never of full row rank $N-1$, which means that the multi-agent
system is not structurally controllable. On the other hand, if the
isolated agent is one follower, we get that the matrix pair
$(A_1+A_2+\ldots+A_m,B_1+B_2+\ldots+B_m)$ is reducible. By
definition \ref{def5}, the controllability matrix
\begin{equation}\begin{split}[
&B_1+B_2+\ldots+B_m,\\&(A_1+A_2+\ldots +A_m)(
B_1+B_2+\ldots+B_m),\\&,\cdots,\\&(A_1+A_2+\ldots+A_m)^{N-2}(
B_1+B_2+\ldots+B_m)]\nonumber
\end{split}\end{equation}
always has at least one row that is identically zero. Expanding the
matrix yields
\begin{equation}\begin{split}[&B_1+B_2+\ldots+B_m,
\\&A_1B_1+A_2B_1+\ldots+A_mB_1+A_1B_2+A_2B_2\\&
+\ldots+A_mB_2+\ldots +A_1B_m+A_2B_m\ldots+A_mB_m\\&,\ldots,\\&
A_1^{N-2}B_1
+A_2A_1^{N-3}B_1+\ldots+A_m^{N-2}B_m].\nonumber\end{split}\end{equation}
The zero row is identically zero for every parameter. This implies
that every component in this matrix, such as $B_i,A_iB_j~and
~A_i^pA_j^qB_r$, has the same row always to be zero. As a result,
the controllability matrix
\begin{equation}\begin{split}\label{eq3} &[B_1,\ldots,B_m,
A_1B_1,\ldots,A_mB_1,\ldots,A_mB_m,A_1^2B_1,\ldots,A_mA_1B_1,\\&\ldots,A_1^2B_m,\ldots,
A_mA_1B_m,\ldots,A_1^{n-1}B_1,\ldots,A_mA_1^{n-2}B_1,\ldots,\\&A_1A_m^{n-2}B_m,\ldots,
A_m^{n-1}B_m] \end{split}\nonumber\end{equation} always has one zero
row. Therefore, the multi-agent system (\ref{eq6}) is not
structurally controllable. Until now, we have got the necessity
proved.

\textit{Sufficiency}: If the union graph $\mathcal{G}$ is connected,
we want to prove that the multi-agent system (\ref{eq6}) 
is structurally controllable.

According to lemma \ref{lem2},  the connectedness of the union graph
$\mathcal{G}$ implies that the corresponding system
$(A_1+A_2+A_3+\ldots +A_m, B_1+B_2+B_3+\ldots +B_m)$ is structurally
controllable. Then there exist some scalars for the parameters in
system matrices that make the controllability matrix
\begin{equation}\begin{split}
&[B_1+B_2+\ldots+B_m,\\
&(A_1+A_2+\ldots+A_m)( B_1+B_2+\ldots+B_m),\\&,\ldots,
\\&(A_1+A_2+\ldots+A_m)^{N-2}( B_1+B_2+\ldots+B_m)],\nonumber\end{split}\end{equation}

has full row rank $N-1$. Expanding the matrix, it follows that the
matrix
\begin{equation}\begin{split}[&B_1+B_2+\ldots+B_m,
\\&A_1B_1+A_2B_1+\ldots+A_mB_1+A_1B_2+A_2B_2\\&
+\ldots+A_mB_2+\ldots +A_1B_m+A_2B_m\ldots+A_mB_m\\&,\ldots,\\&
A_1^{N-2}B_1
+A_2A_1^{N-3}B_1+\ldots+A_m^{N-2}B_m],\nonumber\end{split}\end{equation}
has full rank $N-1$. Next, we add some column vectors to the above
matrix and get
\begin{equation}\begin{split}
&[B_1+B_2+\ldots+B_m,
B_2,\ldots,B_m,\\&A_1B_1+A_2B_1+\ldots+A_mB_1+A_1B_2+A_2B_2+\ldots+A_mB_2\\&+\ldots+
A_1B_m+A_2B_m+\ldots+A_mB_m, A_2B_1,A_3B_1,\ldots,A_mB_m\\&,\ldots
,\\& A_1^{N-2}B_1 +A_2A_1^{N-3}B_1+\ldots+A_m^{N-2}B_m,
A_2A_1^{N-3}B_1,\ldots,
A_m^{N-2}B_m].\nonumber\end{split}\end{equation} This matrix still
have $N-1$ linear independent column vectors, so it has full row
rank. Next, subtract $B_2,\ldots,B_m$ from $ B_1+B_2+\ldots+B_m$;
subtract $A_2B_1,\ldots,A_mB_m$ from $ A_1B_1+A_2B_1+\ldots+A_mB_m $
and subtract $A_2A_1^{N-3}B_1,\ldots, A_m^{N-2}B_m$ from
$A_1^{N-2}B_1 +A_2A_1^{N-3}B_1+\ldots+ A_m^{N-2}B_m$. Since this
column fundamental transformation will not change matrix
rank, the matrix still has full row rank. Now the matrix becomes\\
\begin{equation}\begin{split}\label{eq3} &[B_1,\ldots,B_m,
A_1B_1,\ldots,A_mB_1,\ldots,A_mB_m,A_1^2B_1,\ldots,A_mA_1B_1,\\&\ldots,A_1^2B_m,\ldots,
A_mA_1B_m,\ldots,A_1^{n-1}B_1,\ldots,A_mA_1^{n-2}B_1,\ldots,\\&A_1A_m^{n-2}B_m,\ldots,
A_m^{n-1}B_m] \end{split}\nonumber\end{equation} which is the
controllability matrix of system (\ref{eq6}) and has full row rank
$N-1$. Therefore, the multi-agent system is structurally
controllable.
\end{proof}
\section{Structural Controllability of Multi-agent System with Multi-Leader}
In the above discussion, we assume the multi-agent system has
totally N agents and the $N$-th agent serves as the leader and takes
commands and controls from outside operators directly, while the
rest $N-1$ agents are followers and take controls as the nearest
neighbor law. In the following part, we will discuss the situation
that several agents are chosen as the leaders of the whole
multi-agent systems, which is actually an extension of single leader
case.

Similar to the single leader case, the multi-agent system with
multiple leaders is given by:
\begin{equation}\label{eq17}
\left\{\begin{array}{clc}
   \dot x_i &=& u_i, \quad\quad\quad\hfill i=1,\ldots,N  \\
   \dot x_{N+j} &=& u_{N+j}, \quad\quad\quad \hfill j=1,\ldots,n_l \\
\end{array}\right.
\end{equation}
where $N$ and $n_l$ represent the number of followers and leaders,
respectively. $x_i$ indicates the state of the $i$th agent,
$i=1,\ldots,N+n_l$.

The control strategy $u_i$, $i=1,\ldots,N$ for driving all follower
agents is the same as the single leader case. The leaders' control
signal is still not influenced by the followers and we are allowed
to pick $u_{N+j}$, $j=1,\ldots,n_l$ arbitrarily. For simplicity, we
use vector $x$ to stand for the followers' states and $z$ to stand
for the leaders' states.

Then the whole multi-agent system equipped with $m$ communication
topologies can be written in a compact form
\begin{equation}\label{eq18}
\begin{array}{*{20}c}
   {\left[ {\begin{array}{*{20}c}
   {\dot x}  \\
   {\dot z}  \\
\end{array}} \right] = \left[ {\begin{array}{*{20}c}
   A_i & B_i  \\
   0 & 0  \\
\end{array}} \right]\left[ {\begin{array}{*{20}c}
   x  \\
   z  \\
\end{array}} \right] + \left[ {\begin{array}{*{20}c}
   0  \\
   {u }  \\
\end{array}} \right]} & {}
\end{array},
\end{equation}
or, equivalently,
\begin{equation}\label{eq19}
\left\{\begin{array}{clc}
   \dot x &=& A_i x + B_i z , \hfill  \\
   \dot z &=& u , \hfill  \\
\end{array}\right.
\end{equation}
where $i\ \in \{1,\ldots,m\}$. $A_i \in \mathbb{R}^{N \times N}$ and
$B_i\in \mathbb{R}^{N \times n_l}$ are both sub-matrices of the
corresponding graph Laplacian matrix $-\mathcal{L}$.

The dynamics of the followers can be extracted as
\begin{equation}\label{eq20}
\begin{array}{clc}
   \dot x &=& A_i x + B_i z ,~~ i\in \{1,\ldots,m\}.
  \end{array}
\end{equation}
\begin{remark}\label{rem6}Similar with the single leader case, the structural controllability of system
(\ref{eq18}) coincides with the structural controllability of system
(\ref{eq20}). And we say that the multi-agent system (\ref{eq17})
with switching topology and multi-leader is structurally
controllable if and only if the switched linear system (\ref{eq20})
is structurally controllable with $z$ as the control inputs.
\end{remark}

Before proceeding further, we first discuss the structural
controllability of multi-agent systems with multi-leader under fixed
topology with the following dynamics:
\begin{equation}\label{eq21}
\begin{array}{clc}
   \dot x &=& Ax + Bz, \end{array}
\end{equation}
where $A \in \mathbb{R}^{N \times N}$ and $B\in \mathbb{R}^{N \times
n_l}$ are both sub-matrices of the graph Laplacian matrix
$-\mathcal{L}$.

In \cite{zji}\cite{ZZHZ}, a new graph topology: leader-follower
connected topology was proposed:

\begin{definition}\label{def12}(\cite{zji})
A follower subgraph $\mathcal{G}_f$ of the interconnection graph
$\mathcal{G}$ is the subgraph induced by the follower set
$\mathcal{V}_f$ (here is $x$). Similarly, a leader subgraph
$\mathcal{G}_l$ is the subgraph induced by the leader set
$\mathcal{V}_l$ (here is $z$).
\end{definition}

Denote by $\mathcal{G}_{c_1},\ldots,\mathcal{G}_{c_\gamma},$ the
connected components in the follower $\mathcal{G}_f$. The definition
of leader-follower connected topology is as follows:

\begin{definition}\label{def13}(\cite{ZZHZ})
The interconnection graph $\mathcal{G}$ of multi-agent system
(\ref{eq21}) is said to be leader-follower connected if for each
connected component $\mathcal{G}_{c_i}$ of $\mathcal{G}_f$, there
exists a leader in the leader subgraph $\mathcal{G}_l$, so that
there is an edge between this leader and a node in
$\mathcal{G}_{c_i}, i=1,\ldots,\gamma$.
\end{definition}

Based on this new graph topology, we can derive the criterion for
structural controllability for multi-agent system (\ref{eq21}) under
fixed topology:
\begin{theorem}\label{the4}
The multi-agent system (\ref{eq21}) with multi-leader and fixed
topology under the communication topology $\mathcal{G}$ is
structurally controllable if and only if graph $\mathcal{G}$ is
leader-follower connected.
\end{theorem}
\begin{proof}
\textit{~~Necessity:} The idea of necessity proof is similar to the
proof of lemma 3 in \cite{zji}. We assume that there exists one
connected component $\mathcal{G}_{c_p}$ not connected to the leader
subgraph $\mathcal{G}_l$. Define $A_i$ and $B_i$ matrices as
sub-matrices of $A$ and $B$, the same as the $F_i$ and $R_i$
matrices in lemma 3 of \cite{zji}. Following the analysis in lemma 3
of \cite{zji}, it can be easily got that the controllability matrix
of multi-agent system (\ref{eq21}) is:
\begin{equation}\label{eq22}
  \mathcal{C} = \left[ {\begin{array}{*{20}c}
   B_1 & A_1B_1 & A_1^2B_1 &\cdots& A_1^{N-1}B_1 \\
   \vdots & \vdots & \vdots &\cdots& \vdots \\
   0 & 0 & 0 &\cdots& 0 \\
   \vdots & \vdots & \vdots &\cdots& \vdots \\
   B_\gamma & A_\gamma B_\gamma & A_\gamma^2B_\gamma &\cdots& A_\gamma^{N-1}B_\gamma \\
\end{array}} \right].
\end{equation}
Consequently, rank $\mathcal{C}$ = row rank $\mathcal{C}<N$. The
maximum rank of $\mathcal{C}$ is less than $N$, which implies that
the corresponding multi-agent system (\ref{eq21}) is not
structurally controllable.

\textit{~~Sufficiency:}We use the proof of theorem 1 in \cite{ZZHZ}
to help us prove the sufficiency. The communication graph
$\mathcal{G}$ consists of several connected components
$\mathcal{G}^{(i)}$, $i=1,\ldots,\kappaup$, which can be partitioned
into two subgraphs: induced leader subgraph $\mathcal{G}^{(i)}_l$
and induced follower subgraph $\mathcal{G}^{(i)}_f$. For each
connected components $\mathcal{G}^{(i)}$, $i=1,\ldots,\kappaup$, it
can be modeled as a linear system with its system matrices being
sub-matrices of $A$ and $B$ matrices. Following the analysis in
theorem 1 in \cite{ZZHZ}, we can get the following equation:
\begin{equation}\label{eq23}
rank~\mathcal{C}= rank~ \mathcal{C}_1+rank
~\mathcal{C}_2+\ldots+rank~\mathcal{C}_\kappaup,
\end{equation}
where $\mathcal{C}$ is the controllability matrix of multi-agent
system (\ref{eq21}) and $\mathcal{C}_i$ is the controllability
matrix of connected component $\mathcal{G}^{(i)}$. The independence
of these connected components guarantees the independence of free
parameters in the corresponding matrices, which correspond to the
communication weights of the linkages. Consequently, we have that
\begin{center}\label{eq24}
$g$-$rank$ $\mathcal{C}$ = $g$-$rank$ $\mathcal{C}_1$+ $g$-$rank$
$\mathcal{C}_2$+$\ldots$+ $g$-$rank$ $\mathcal{C}_\kappaup$
\end{center}
where $g$-rank of a structured matrix $M$ is defined to be the
maximal rank that $M$ achieves as a function of its free parameters.
Besides, if in some connected component $\mathcal{G}^{(i)}$, there
is more than one leaders, we can split it into several connected
components with single leader or choose one as leader and set all
weights of the communication linkages between the followers and
other leaders to be zero. After doing this, connected component
$\mathcal{G}^{(i)}$ is a connected topology with single leader.
According to lemma \ref{lem2}, $\mathcal{C}_i$ has full $g$-$rank$,
which equals to the number of follower agents in
$\mathcal{G}^{(i)}$. Moreover, there is no common follower agent
among the connected components. Consequently, $g$-$rank$
$\mathcal{C}$=$N$ and multi-agent system (\ref{eq21}) is
structurally controllable.
\end{proof}

With the above definitions and theorems, we are in the position to
present the graphical interpretation of structural controllability
of multi-agent systems under switching topology with multi-leader:
\begin{theorem}\label{the5}
The multi-agent system (\ref{eq18}) or (\ref{eq20}) with the
communication topologies $\mathcal{G}_i$, $i\in \{1,\ldots,m\}$ and
multi-leader is structurally controllable if and only if the union
graph $\mathcal{G}$ is leader-follower connected.
\end{theorem}
\begin{proof}
As stated in remark \ref{rem3}, the union graph $\mathcal{G}$ is the
representation of the linear system:
$(A_1+A_2+A_3+\ldots+A_m,B_1+B_2+B_3+\ldots+B_m)$. Therefore, the
condition that the union graph $\mathcal{G}$ is leader-follower
connected is equivalent to the condition that linear system
$(A_1+A_2+A_3+\ldots+A_m,B_1+B_2+B_3+\ldots+B_m)$ is structurally
controllable. Following the proof steps in theorem \ref{the1}, this
result can get proved.
\end{proof}


\section{ Numerical Examples}

Next we will give two examples to illustrate the results in this
paper and for simplicity, we take single leader case as examples.

\vspace{0.8in} \setlength{\unitlength}{0.0215in}
\begin{picture}(80,35)

\put(10,20){\line(0,1){30}}\put(40,20){\line(0,1){30}}
\put(60,20){\line(1,0){30}}\put(60,50){\line(1,0){30}}
\put(110,20){\line(0,1){30}}\put(140,50){\line(-1,0){30}}
\put(110,20){\line(1,0){30}}\put(140,50){\line(0,-1){30}}

\put(10,20){\circle*{2}}\put(40,20){\circle*{2}}
\put(60,20){\circle*{2}}\put(90,20){\circle*{2}}
\put(110,20){\circle*{2}}\put(140,20){\circle*{2}}
\put(10,50){\circle*{2}}\put(40,50){\circle*{2}}
\put(60,50){\circle*{2}}\put(90,50){\circle*{2}}
\put(110,50){\circle*{2}}\put(140,50){\circle*{2}}

\put(7,20){\makebox(0,0)[c]{$3$}}\put(57,20){\makebox(0,0)[c]{$3$}}\put(107,20){\makebox(0,0)[c]{$3$}}
\put(43,20){\makebox(0,0)[c]{$0$}}\put(93,20){\makebox(0,0)[c]{$0$}}\put(143,20){\makebox(0,0)[c]{$0$}}
\put(7,50){\makebox(0,0)[c]{$1$}}\put(57,50){\makebox(0,0)[c]{$1$}}\put(107,50){\makebox(0,0)[c]{$1$}}
\put(43,50){\makebox(0,0)[c]{$2$}}\put(93,50){\makebox(0,0)[c]{$2$}}\put(143,50){\makebox(0,0)[c]{$2$}}

\put(25,15){\makebox(-2,0){$(a)$}}
\put(75,15){\makebox(-2,0){$(b)$}}
\put(125,15){\makebox(-2,0){$(c)$}}

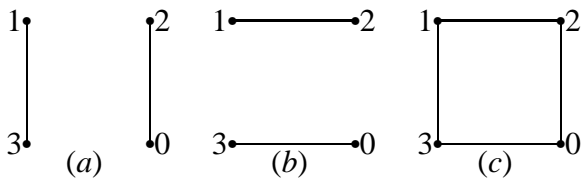
\put(75,0){\makebox(15,1)[c]{{\footnotesize Fig. 1. Switched network
with two subsystems }}}

\end{picture}
\vspace{0.2in}

We consider here a four-agent network with agent 0 as the leader and
with switching topology described by the graphs in Fig. 1(a)-(b)
(the self-loops are not depicted, because it will not influence the
connectivity). Overlaying the subgraphs together can get the union
graph $\mathcal{G}$ of this example as shown in Fig. 1(c). It turns
out that the union graph of the switched system is connected. By
theorem \ref{the1}, it is clear that the multi-agent system is
structurally controllable.

Next, the rank condition of this multi-agent system will be checked.

From Fig. 1, calculating the Laplacian matrix for each subgraph
topology, it can be obtained that the system matrices of each
subsystem are (one thing we should mention here with the control
strategy that each agent can use its own state information, the
diagonal elements always have free parameters, so we can get the
following form of sub-matrix of Laplacian matrix) :
\begin{equation}\label{124}
A_1= \left[ {\begin{array}{*{20}c}
   \lambda_1 & 0&\lambda_4  \\
   0&\lambda_2 & 0  \\
   \lambda_5&0&\lambda_3
\end{array}} \right],~~B_1=\left[ {\begin{array}{*{20}c}
   0 \\
   \lambda_6  \\
   0
\end{array}} \right];
A_2= \left[ {\begin{array}{*{20}c}
   \lambda_7& \lambda_{10}& 0 \\
   \lambda_{11}&\lambda_8 & 0  \\
   0&0&\lambda_9
\end{array}} \right],~~B_2=\left[ {\begin{array}{*{20}c}
   0 \\
   0  \\
   \lambda_{12}
\end{array}} \right]. \nonumber \end{equation}

According to lemma \ref{lem1}, the controllability matrix for this
switched linear system
is:$\newline[B_1,B_2,A_1B_1,A_2B_1,A_1B_2,A_2B_2,A_1^2B_1,A_2A_1B_1,A_1A_2B_1,\\A_2^2B_1,A_1^2B_2,A_2A_1B_2,A_1A_2B_2,A_2^2B_2].$
\newline After simple calculation, we can find three column vectors
in the controllability matrix:
\begin{equation}\label{124}
\left[ {\begin{array}{*{20}c}
  0  \\
   \lambda_6  \\
   0
   \end{array}} \right],~~\left[ {\begin{array}{*{20}c}
   0 \\
   0\\
   \lambda_{12} \\

\end{array}} \right],~~\left[ {\begin{array}{*{20}c}
   \lambda_4\lambda_{12} \\
   0\\
   \lambda_3\lambda_{12} \\

\end{array}} \right].\nonumber \end{equation}

Imposing all the parameters scalar 1, it follows that these three
column vectors are linearly independent and this controllability
matrix has full row rank. Therefore, the multi-agent system is
structurally controllable.


In the second example, we still consider a four-agent network with
agent 0 as the leader and with switching topology described by the
graphs in Fig. 2(a)-(b). Overlaying the subgraphs together can get
the union graph $\mathcal{G}$ of this example shown in Fig. 2(c). It
turns out that the union graph of the switched system is
disconnected, because agent 2 is isolated. According to theorem
\ref{the1}, it is clear that the multi-agent system is not
structurally controllable.

\vspace{0.8in} \setlength{\unitlength}{0.0215in}

\begin{picture}(80,35)

\put(10,20){\line(0,1){30}}\put(40,20){\line(-1,1){30}}
\put(60,20){\line(1,0){30}}\put(60,50){\line(0,-1){30}}
\put(110,20){\line(0,1){30}}
\put(110,20){\line(1,0){30}}\put(140,20){\line(-1,1){30}}

\put(10,20){\circle*{2}}\put(40,20){\circle*{2}}
\put(60,20){\circle*{2}}\put(90,20){\circle*{2}}
\put(110,20){\circle*{2}}\put(140,20){\circle*{2}}
\put(10,50){\circle*{2}}\put(40,50){\circle*{2}}
\put(60,50){\circle*{2}}\put(90,50){\circle*{2}}
\put(110,50){\circle*{2}}\put(140,50){\circle*{2}}

\put(7,20){\makebox(0,0)[c]{$3$}}\put(57,20){\makebox(0,0)[c]{$3$}}\put(107,20){\makebox(0,0)[c]{$3$}}
\put(43,20){\makebox(0,0)[c]{$0$}}\put(93,20){\makebox(0,0)[c]{$0$}}\put(143,20){\makebox(0,0)[c]{$0$}}
\put(7,50){\makebox(0,0)[c]{$1$}}\put(57,50){\makebox(0,0)[c]{$1$}}\put(107,50){\makebox(0,0)[c]{$1$}}
\put(43,50){\makebox(0,0)[c]{$2$}}\put(93,50){\makebox(0,0)[c]{$2$}}\put(143,50){\makebox(0,0)[c]{$2$}}

\put(25,15){\makebox(-2,0){$(a)$}}
\put(75,15){\makebox(-2,0){$(b)$}}
\put(125,15){\makebox(-2,0){$(c)$}}

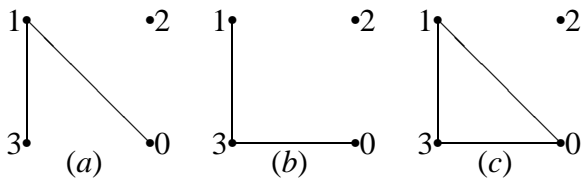
\put(75,0){\makebox(15,1)[c]{{\footnotesize Fig. 2. Another switched
network with two subsystems }}}

\end{picture}
\vspace{0.2in}

Similarly, the rank condition of this switched linear system needs
to be checked to see whether it is structurally controllable or not.

From Fig. 2,  calculating the Laplacian matrix for each graphic
topology, it is clear that the system matrices of each subsystem are
:
\begin{equation}\label{124}
A_1= \left[ {\begin{array}{*{20}c}
   \lambda_1 & 0&\lambda_4  \\
   0&\lambda_2 & 0  \\
   \lambda_5&0&\lambda_3
\end{array}} \right],~~B_1=\left[ {\begin{array}{*{20}c}
   \lambda_6  \\
   0 \\
   0
\end{array}} \right];
A_2= \left[ {\begin{array}{*{20}c}
   \lambda_7&0 & \lambda_{10} \\
   0&\lambda_8 & 0  \\
   \lambda_{11}&0&\lambda_9
\end{array}} \right],~~B_2=\left[ {\begin{array}{*{20}c}
   0 \\
   0  \\
   \lambda_{12}
\end{array}} \right].\nonumber  \end{equation}

Computing the controllability matrix of this example yields the
controllability matrix:
\begin{equation}
\left[{\begin{array}{*{20}c}
\lambda_6&0&\lambda_1\lambda_6&
\ldots&\lambda_7\lambda_{10}\lambda_{12}+\lambda_9\lambda_{10}\lambda_{12}\\
0&0&0
&\ldots&0\\
0&\lambda_{12}&\lambda_5\lambda_6&
\ldots&\lambda_{10}\lambda_{11}\lambda_{12}+\lambda_9^2\lambda_{12}
\end{array}}\right].\nonumber
\end{equation}

This matrix has the second row always to be zero for all the
parameter values, which makes the maximum rank of this matrix less
than 3. Therefore, this multi-agent system is not structurally
controllable.


\section{Conclusions and Future Work}

In this paper, the structural controllability problem of the
multi-agent systems interconnected via a switching weighted topology
has been considered. Based on known results in the literature of
switched systems and graph theory, graphic necessary and sufficient
conditions for the structural controllability of multi-agent systems
under switching communication topologies were derived. It was shown
that the multi-agent system is structurally controllable if and only
if the union graph $\mathcal{G}$ is connected (single leader) or
leader-follower connected (multi-leader). The graphic
characterizations show a clear relationship between the
controllability and interconnection topologies and give us a
foundation to design the optimal control effect for the switched
multi-agent system.

Some interesting remarks can be made on this result. First, it gives
us a clear understanding on what are the necessary information
exchanges among agents to make the group of agents behavior in a
desirable way. Second, it provides us a guideline to design
communication protocols among dynamical agents. It is required that
the resulted communication topology among agents should somehow
remain connected as time goes on, which is quite intuitive and
reasonable. Third, it is possible to reduce communication load by
disable certain linkages or make them on and off as long as the
union graph is connected. Several interesting research questions
arise from this scenario. For example, what is the optimal switching
sequence of  topologies in the sense of minimum communication cost?
How to co-design the switching topology path and control signals to
achieve desirable configuration in an optimal way? We will
investigate these questions in our future research.

\bibliographystyle{plain}        

\end{document}